\documentclass[12pt]{article}

\usepackage{graphicx}
\usepackage{amsmath}
\usepackage{amssymb}
\usepackage{amsthm}
\newcommand{\be}{\begin{equation}}
\newcommand{\ee}{\end{equation}}
\newcommand{\ba}{\begin{array}}
\newcommand{\ea}{\end{array}}
\newcommand{\bea}{\begin{eqnarray}}
\newcommand{\eea}{\end{eqnarray}}

\newcommand{\calH}{{\cal H }}

\newcommand{\calZ}{{\cal Z }}

\newcommand{\calS}{{\cal S }}

\newcommand{\calD}{{\cal D }}
\newcommand{\calY}{{\cal Y }}
\newcommand{\calV}{{\cal V }}

\newcommand{\ZZ}{\mathbb{Z}}

\newcommand{\RR}{\mathbb{R}}

\newcommand{\la}{\langle}
\newcommand{\ra}{\rangle}

\newcommand{\nn}{\nonumber}

\newtheorem{dfn}{Definition}
\newtheorem{lemma}{Lemma}
\newtheorem{prop}{Proposition}
\newtheorem{theorem}{Theorem}

\newtheorem{corollary}{Corollary}

\title{A short proof of stability of topological order under local perturbations}

\author{Sergey Bravyi\footnote{IBM Watson Research Center, Yorktown Heights NY 10594 (USA); sbravyi@us.ibm.com},  \, \,
and  \, \,
Matthew B. Hastings\footnote{Microsoft Research Station Q, CNSI Building, University of California, Santa Barbara, CA,
93106 (USA); mahastin@microsoft.com},
}

\begin{document}

\maketitle

\abstract{Recently, the stability of certain topological phases of matter under weak perturbations was proven.
Here, we present a short, alternate proof of the same result.
We consider models of topological quantum order for which the unperturbed Hamiltonian $H_0$ can be written as a sum of local pairwise
commuting projectors on a $D$-dimensional lattice.  We consider
a perturbed Hamiltonian $H=H_0+V$ involving a generic perturbation $V$
that can be written as a sum of short-range bounded-norm interactions.
We prove that if the strength of $V$ is below a constant threshold value then $H$
has well-defined spectral bands originating from the low-lying
eigenvalues of $H_0$. These bands are separated from the rest of the spectrum and from each other by a constant gap.
The width of the band originating from the smallest eigenvalue of $H_0$
decays faster than any power of the lattice size.
 }

\section{Introduction}
Quantum spin Hamiltonians exhibiting topological quantum order (TQO) have a remarkable property that 
their ground state degeneracy cannot be lifted by generic local perturbations~\cite{Wen90,tc}. In addition, the spectral gap above the ground state does not close in a presence of such perturbations.
 This property is in sharp contrast with the behavior of classical spin Hamiltonians such as the 2D Ising model for which the ground state degeneracy is unstable in a presence of the external magnetic field.
Recently the authors of~\cite{prev} presented a rigorous proof of the gap stability for a large class of Hamiltonians
describing TQO.  The proof was direct, but very long.
This paper presents a short technical note giving an alternate proof of the same result.  We feel that both proofs are worth  having,
since they are complementary.  In particular, the present proof is much  shorter, but less direct.  Since we consider the same
models and prove the same result, many of the definitions in this paper are identical.

The interest in proving stability of a  spectral gap under perturbations has
increased due to recent progress in mathematical physics, where a combination of
Lieb-Robinson bounds~\cite{lr1,lr2,lr3} and the method of quasi-adiabatic continuation~\cite{hwen,tjo}
 with appropriately chosen filter functions now provides powerful techniques for studying
the properties of gapped local quantum Hamiltonians.
Perhaps surprisingly, we will use similar techniques to
prove the existence of a gap.
In contrast, to the best of our knowledge, all previous works addressing the gap stability problem
relied on perturbative methods such as the cluster expansion for the thermal Gibbs state~\cite{KenTas92,yarot,klich},
Kirkwood-Thomas expansion~\cite{KT83,Datta02}, or the coupled cluster method~\cite{Yarotsky04,Bravyi08}.
Our techniques are particularly well suited to handle topologically ordered ground states, although they can also be applied to topological trivial states such as ground states of classical Hamiltonians (in the classical case stability is possible only for
non-degenerate ground states). Further, as shown in \cite{prev,hwen}, the gap stability implies that many of the topological properties of the unperturbed system carry over to the perturbed system.

We consider  a system composed of  finite-dimensional quantum particles (qudits)
occupying sites of a $D$-dimensional lattice $\Lambda$ of linear size $L$.
Suppose the unperturbed Hamiltonian $H_0$  can be written as a sum of
geometrically local pairwise commuting projectors,
\[
H_0=\sum_{A \subseteq \Lambda} Q_A,
\]
such that its ground subspace $P$ is annihilated by every projector, $Q_AP=0$.
We  impose two extra conditions on $H_0$ and the ground subspace $P$
which are responsible for the
topological order.  In~\cite{prev}, it was shown that 
neither of 
these conditions by itself is sufficient for the gap stability.
Let us first state these conditions informally (see Section~\ref{sec:tqo} for formal definitions):
\begin{enumerate}
\item[] {\bf TQO-1:} The ground subspace $P$ is a quantum  code with a macroscopic distance,
\item[] {\bf TQO-2:} Local ground subspaces are consistent with  the global one.
\end{enumerate}

We consider a perturbation $V$ that can be written as a sum of local interactions
\[
V=\sum_{r\ge 1}\; \;  \sum_{A\in \calS(r)}  \; \; V_{r,A},
\]
where $\calS(r)$ is a set of cubes of linear size $r$ and $V_{r,A}$ is an operator acting on sites of $A$.
We assume that the magnitude of interactions decays exponentially for large $r$,
\[
\max_{A\in \calS(r)} \|V_{r,A}\|\le Je^{-\mu r},
\]
where $J,\mu>0$ are some constants independent of $L$.
Our main result is the following theorem.
\begin{theorem}
\label{thm:main}
There exist constants $J_0,c_1>0$ depending only on $\mu$
and the spatial dimension $D$ such that for all $J\le J_0$ the spectrum of $H_0+V$ is
contained (up to an overall energy shift) in the union of intervals $\bigcup_{k\ge 0} I_k$, where
$k$ runs over the spectrum of $H_0$ and $I_k$ is the closed interval
\[
I_k=[k(1-c_1J)-\delta,k(1+c_1J)+\delta],
\]
for some $\delta$ bounded by $J$ times a quantity decaying faster than any power of $L$.
\end{theorem}
Using the quasi-adiabatic continuation operators in \cite{inprep}, this superpolynomial bound on $\delta$ can be improved
to a bound by an exponential of a polynomial of $L$, where the polynomial may be taken to be arbitrarily close to linear.

The stability proof presented in~\cite{prev} involved two main steps: (i) proving stability for a special class
of {\em block-diagonal} perturbations $V$ composed of local interactions $V_{r,A}$ preserving the ground subspace of $H_0$,
and (ii) reducing generic local perturbations to block-diagonal perturbations. The most technical part of the proof was step~(ii)
which required complicated convergence analysis for Hamiltonian flow equations.

In the present paper we show how to simplify step~(ii) significantly using  an exact quasi-adiabatic continuation technique~\cite{hwen,tjo}.
In Section~\ref{sec:LR} we define two types of block-diagonal perturbations:
 the ones that preserve the ground subspace of $H_0$ locally ($[V_{r,A},P]=0$) and globally ($[V,P]=0$). We prove stability under locally block-diagonal perturbations in Section~\ref{sec:lbd}
which mostly follows~\cite{prev} and uses the technique of relatively bounded operators.
The reduction from globally block-diagonal to locally block-diagonal perturbations is in two steps;
in Lemma~\ref{rewrite1} we
exploit the idea of~\cite{solvloc,kitaev} to write a Hamiltonian
of a gapped system as a sum of terms such that the ground states are
(approximate) eigenvectors of each term separately, and then we show that such terms can be written as locally
block-diagonal perturbations in Section~\ref{sec:red}.
Finally, we use an exact quasi-adiabatic continuation technique presented in Section~\ref{sec:continue} to reduce generic local perturbations to
globally block-diagonal perturbations in Section~\ref{sec:proof}.

In contrast to~\cite{prev} we use certain self-consistent assumptions on the spectral gap
of the perturbed Hamiltonian $H_s=H_0+sV$. In Section~\ref{sec:proof} we prove that if the minimal gap $\Delta_{min}$
along the path $0\le s\le 1$ is bounded by some constant (say $1/2$) then,
in fact, $\Delta_{min}$ is bounded by a larger constant (say $3/4$). We use continuity arguments to translate this
result into an unconditional constant bound on $\Delta_{min}$, see Section~\ref{sec:proof}.  The continuity arguments
we use rely on the fact that the spectrum of $H_0+sV$ is a continuous function of $s$ for any finite sized system;
however, our theorem gives bounds that are independent of system size.

\section{Hamiltonians describing TQO}
\label{sec:tqo}

To simplify notation we shall restrict ourselves to the spatial dimension $D=2$.
A generalization to an arbitrary $D$ is straightforward.
Let $\Lambda=\ZZ_L\times \ZZ_L$ be a two-dimensional square lattice of linear size $L$ with periodic boundary conditions.
We assume that every site $u\in \Lambda$ is occupied by a finite-dimensional quantum particle (qudit)
such that the Hilbert space describing $\Lambda$ is a tensor product
\be
\calH=\bigotimes_{u\in \Lambda} \calH_u, \quad \dim{\calH_u}=O(1).
\ee
Let $\calS(r)$ be a set of all square blocks $A\subseteq \Lambda$ of size $r\times r$,
where $r$ is a positive integer.
Note that $\calS(r)$ contains $L^2$  translations of some elementary square
of size $r\times r$ for all $r<L$, $\calS(L)=\Lambda$, and $\calS(r)=\emptyset$ for $r>L$.
We shall assume that the unperturbed Hamiltonian $H_0$ involves only $2\times 2$ interactions
(otherwise consider a coarse-grained lattice):
\be
\label{H_0}
H_0=\sum_{A\in \calS(2)} Q_A.
\ee
Here $Q_A$ is an interaction that has support on a square $A$.
We also assume that  the interactions $Q_A$ are pairwise commuting  projectors,
\be
Q_A^2=Q_A, \quad Q_A Q_B = Q_B Q_A \quad \mbox{for all $A,B\in \calS(2)$}.
\ee
Well known examples of Hamiltonians composed of commuting projectors are
Levin-Wen  models~\cite{LW} and quantum double models~\cite{tc}.

We assume that $H_0$ has zero ground state energy, i.e.,
the ground states of $H_0$ are zero-eigenvectors of every projector $Q_A$.
Let $P$ and $Q$ be the projectors onto the ground subspace and the excited subspace of $H_0$,
that is,
\be
P=\prod_{A\in \calS(2)} (I-Q_A) \quad \mbox{and} \quad Q=I-P.
\ee
For any square $B\in \calS(r)$, $r\ge 2$ define local versions of $P$ and $Q$ as
\be
\label{local_projectors}
P_B =\prod_{\substack{A\in \calS(2) \\ A\subseteq B}} (I-Q_A) \quad \mbox{and} \quad Q_B=I-P_B.
\ee
Note that $P_B$ and $Q_B$ have support on  $B$.

We shall need two extra properties of $H_0$ related to TQO defined in \cite{prev}.
We shall assume that there exists an integer $L^*\ge L^a$ for some constant $a>0$ such that
one has the following properties:
\begin{enumerate}
\item
{\bf TQO-1:} Let $A\in \calS(r)$ be any square of size $r\le L^*$. Let $O_A$
be any operator acting on $A$. Then
\[
PO_AP =c P
\]
for some complex number $c$.

\item {\bf TQO-2:} Let $A\in \calS(r)$ be any square of size $r\le L^*$
and let $B\in \calS(r+2)$ be the square that contains $A$ and all nearest neighbors of $A$.
 Let $O_A$
be any operator acting on $A$ such that $O_A P=0$. Then
\[
O_A P_B=0.
\]
\end{enumerate}

The property TQO-1 is often taken as a definition of TQO, see Ref.~\cite{Wen90,tc,Freedman01,bhv}.
One can also show that $L^*\ge d/4$ where $d$ is the distance of $P$
considered as a codespace of a quantum error correcting code, i.e., the smallest integer $m$ such that erasure
of any subset of $m$ particles can be corrected, see  Ref.~\cite{BPT09} for details.

We will need later a corollary of TQO-1 and TQO-2.
We use $b_l(A)$ to denote the square containing square $A$ as well as all first, second,...,$l$-th neighbors of square $A$.
That is, it is a ``ball" of distance $l$ around $A$.
\begin{corollary}
\label{localrho}
Let $O_A$ be any operator supported on a square $A$.
Let $C=b_2(A)$ and suppose
that $C$ has size bounded by $L^*$.  Then
\be
P_C O_A P_C=c P_C,
\ee
for some constant $c$.

Thus, all states $\psi$ with $P_C \psi=\psi$ have the same reduced density matrix on $A$,
and in particular the reduced density matrix of $\psi$ on square $A$ is equal to the
reduced density matrix of any ground state on $A$.

Finally, we have the useful equality:
\be
\label{useeq}
\|O_A P\| = \|O_A P_C\|.
\ee
\begin{proof}
Let $B=b_1(A)$.
Note that $P_B O_A P_B$ commutes with the Hamiltonian, and hence commutes with $P$.
So, $P \left( P_B O_A P_B \right) P=\left( P_B O_A P_B \right) P$.  So by condition TQO-1, $\left( P_B O_A P_B \right) P=c P$ for some $c$.
Define $O'_A=O_A-c$.  Then, $\left( P_B O'_A P_B \right) P=0$.  So, by condition TQO-2, $\left( P_B O'_A P_B \right) P_C=0$.  Note that
$P_B O'_A P_B$ commutes with $P_C$.  So, $P_C \left( P_B O'_A P_B \right) P_C=0$.  Hence, $P_C O'_A P_C=0$.  So, $P_C O_A P_C=c P_C$.

Now consider the second claim.  Let $\psi$ be any state such that $P_C \psi=\psi$.  Let $\rho_A(\psi)$ be the reduced density
matrix of $\psi$ on $A$.  Then, ${\rm tr}(\rho_A O_A)={\rm tr}(P_C |\psi\rangle\langle\psi| P_C O_A)=c(O_A) {\rm tr}(|\psi\rangle\langle\psi|)=c(O_A)$, where the constant $c(O_A)$ depends only on $O_A$ and not on $\psi$.  Thus, the reduced density matrix is the same
for all vectors such that $P_C\psi=\psi$.

Finally, let $\psi$, $\Psi_0$ satisfy $\|O_A P_C\| = |O_A \psi|$ and $\|O_A P\|=|O_A \Psi_0|$. Then, we have that $P_C \psi = \psi$ and $P\Psi_0 = \Psi_0$, which implies
${\rm tr}(|O_A|^2 |\psi\ra\la\psi|) = {\rm tr}(|O_A|^2 |\Psi_0\ra\la\Psi_0|)$, since their reduced density matrices agree on $A$. Hence, $\|O_A P\| = \|O_A P_C\|$, as claimed.

\end{proof}
\end{corollary}

\section{Local Hamiltonians and the Lieb-Robinson bounds}
\label{sec:LR}
Throughout this paper we restrict ourselves to Hamiltonians with a sufficiently fast spatial decay of
interactions. Any such  Hamiltonian $V$ will be specified using a decomposition into local interactions,
\be
\label{dec}
V=\sum_{r\ge 1} \; \; \sum_{A\in \calS(r)} \; \; V_{r,A},
\ee
where $V_{r,A}^\dag=V_{r,A}$ is an operator
acting non-trivially only on a square $A$. We shall often identify a Hamiltonian and the corresponding decomposition
unless it may lead to confusions. Let us define several important classes of Hamiltonians.
\begin{dfn}
A Hamiltonian $V$ has support near a site $u\in \Lambda$ if all squares in the decomposition Eq.~(\ref{dec})
contain $u$.
\end{dfn}
\begin{dfn}
A Hamiltonian $V$ is globally block-diagonal iff it preserves the ground subspace $P$, that is, $[V,P]=0$.
A Hamiltonian $V$ is locally block-diagonal iff all terms $V_{r,A}$ in the decomposition Eq.~(\ref{dec})
preserve the ground subspace, that is, $[V_{r,A},P]=0$ for all $r,A$.
\end{dfn}

\begin{dfn}
A Hamiltonian $V$ has strength $J$ if there exists a function $f\, : \, \ZZ_+ \to [0,1]$ decaying faster than any power such that
\[
\| V_{r,A} \| \le J\, f(r) \quad \mbox{for all $r\ge 1$ for all $A\in \calS(r)$}.
\]
\end{dfn}
Our main technical tool will be the following corollary of the Lieb-Robinson bound for systems with interactions
decaying slower than exponential~\cite{hastings-koma,bhv}.
\begin{lemma}
\label{lemma:LRslow}
Let $H$ be a  (time-dependent) Hamiltonian with strength $O(1)$. Let $V$ be a Hamiltonian
with strength $J$. Let $U(t)$ be the unitary evolution for time $t$, $|t|\le 1$, generated by $H$.
 Define
\[
\tilde{V} = U(t) V U(t)^\dag.
\]
Then $\tilde{V}$ has strength $O(J)$. If $V$ has support near a site $u$ then $\tilde{V}$ also has support near $u$.
\end{lemma}

We will also need an analogue of Lemma~\ref{lemma:LRslow} for an infinite evolution time.
It will be applicable only to evolution under Hamiltonians $H$ having a finite Lieb-Robinson
velocity~\cite{lr1,lr2,lr3,hastings-koma}, for example, Hamiltonians with exponentially decaying interactions.
\begin{dfn}
A Hamiltonian $V$ has strength $J$ and decay rate $\mu$ iff
\[
\| V_{r,A} \| \le J\, \exp{(-\mu r)} \quad \mbox{for all $r\ge 1$ for all $A\in \calS(r)$}.
\]
\end{dfn}
\begin{lemma}
\label{lemma:LRfast}
Let $H$ be a  (time-dependent) Hamiltonian with strength $O(1)$ and decay rate $\mu>0$. Let $V$ be a Hamiltonian
with strength $J$. Let $U(t)$ be the unitary evolution for time $t$ generated by $H$.
Finally, let $g(t)$ be any function decaying faster than any power for large $|t|$. Define
\[
\tilde{V} = \int_{-\infty}^\infty dt \, g(t) U(t) V U(t)^\dag.
\]
Then $\tilde{V}$ has strength $O(J)$.
 If $V$ has support near a site $u$ then $\tilde{V}$ also has support near $u$.
\end{lemma}
The superpolynomially decaying function of $r$ bounding the norm of $\tilde V_{r,A}$ depends on the superpolynomially decaying function
of $r$ bounding the norm of  $V_{r,A}$ as well as on $\mu$ and the function $g(t)$.
({\em Remark:} In fact, one can show that there exist
$g(t)$ that have the  properties needed later and that decay as an
exponential of a power of $t$, making it  possible also to consider
Hamiltonians $H$ that do not have a decay constant $\mu>0$ but that
instead have sufficiently fast stretched exponential decay.  We omit
this for simplicity.)

\section{Reduction from global to local block-diagonality}
\label{sec:red}
In this section we prove that a certain class of globally block-diagonal perturbations
can be reduced to locally block-diagonal perturbations with a small error by simply rewriting the
Hamiltonian in a different  form.
\begin{lemma}
\label{lemma:solveloc}
Consider a Hamiltonian
\be
H=H_0+\sum_{u\in \Lambda} X_u,
\ee
where $H_0$ obeys TQO-1,2 and $X_u$ is a perturbation with strength $J$
that has  support near $u$. Suppose that
\[
[X_u,P]=0 \quad \mbox{for all $u$}.
\]
Then we can rewrite
\be
H=H_0+V'+\Delta,
\ee
where  $V'$ is a locally block-diagonal perturbation with strength $O(J)$,
and $\|\Delta\|$ decays faster than any power of $L^*$.
\end{lemma}
\begin{proof}
Consider some fixed site $u$ and let $X\equiv X_u$.
By assumption, $X$ has a decomposition $X=\sum_{r\ge 1} X(r)$, where $X(r)=\sum_{A\in \calS(r), \; A\ni u} \; \; X_{r,A}$
and the norm $\|X(r)\|$ decays faster than any power of $r$.
By adding constants to the different terms $X_{r,A}$ and using TQO-1 we can achieve $PX_{r,A}P=0$
for all $r\le L^*$.
Define $\Delta=PX=PXP=\sum_{r>L^*} \; PX(r) P$. Note that $\|\Delta\|$ decays faster than any power of $L^*$.
Define also $X'=X-\Delta$ such that
\[
X'P=0.
\]
We can assume that $X'$ has strength $O(J)$ and its support is near $u$ if we treat $\Delta$
as a single interaction on a square of size $L$. To simplify notation we set $X=X'$ in the rest of the proof.

Choosing any $l\le L^*$ and
applying Eq.~(\ref{useeq}) to $O=\sum_{r=1}^l X(r)$  we arrive at
\be
\label{almosttqo}
\Vert \sum_{r=1}^{l} X(r) P_{b_{l+2}(u)} \Vert =
\Vert \sum_{r=1}^{l} X(r) P \Vert
\le \| X P\| + \|  \sum_{r>l}  X(r) P \Vert \le
 J f(l)
\ee
for some function $f$ decaying faster than any power of $l$.

Consider a given site $u$.  
Let $B_r$ be the union of all squares of size $r$ that contain $u$, such that $X(r)$ 
has support on $B_r$. We have $B_1\subset B_2 \subset \ldots \subset B_M=\Lambda$
for some integer $M$.
Define an orthogonal unity decomposition $I=\sum_{m=1}^{M+1} E_m$ by
\bea
E_1&=&Q_{B_1}, \nn \\
E_m&=& Q_{B_{m}} P_{B_{m-1}} \quad \mbox{for $2\le m\le M$}, \nn \\
E_{M+1}&=&P_{B_{M}}=P.
\eea
Taking into account that $E_{M+1}X=XE_{M+1}=0$ we arrive at
\be
\label{Xdec}
X=\sum_{1\leq p,q,r\le M} E_p X(q) E_r= \sum_{j \geq 1} Y(j)+ \sum_{q\ge 1} Z(q),
\ee
where
\be
Y(j)=\sum_{\substack{1\leq p,r\le M\\ p+r=j}} E_p \left( \sum_{q=1}^{{\rm max}(p,r)-2} X(q) \right) E_r.
\ee
and
\be
Z(q)=\sum_{1\leq p,r\le q+1} E_p X(q) E_r.
\ee
All operators $Y(j)$ and $Z(j)$ are hermitian, annihilate $P$,
and  act  only on sites in the square $B_{j+1}$.
The norm of $Z(q)$ decays
 faster than any power of  $q$ because of the decay of the norm of $X(q)$.

We claim that
the norm of $Y(j)$ decays faster than any power of $j$.
Indeed,
 $Y(j)$ is a sum over $j-1$ different terms corresponding to different
choices of $p,r$.  We show
that the norm of each such term  $p,r$ decays fast enough.  Assume without loss of generality that $p\geq r$.  Then,
$\Vert E_p \Bigl( \sum_{q=1}^{{\rm max}(p,r)-2} X(q) \Bigr) E_r \Vert \leq \Vert E_p \Bigl( \sum_{q=1}^{{\rm max}(p,r)-2} X(q) \Bigr)\Vert$, which decays faster than any power of $p$
 by Eq.~(\ref{almosttqo}). Since $p\ge j/2$,  it decays faster than any power of $j$.

Thus, we have decomposed $X$ as a sum of terms which annihilate $P$, supported on squares of increasing radius about site
$u$, with norm
decreasing faster than any power of the size of the square.
This completes the proof.
\end{proof}

\section{Stability under locally block-diagonal perturbation}
\label{sec:lbd}

In this section, we define the concept of relatively bounded perturbations (see Chapter~IV of~\cite{Kato}) and 
show that the spectral gaps separating low-lying eigenvalues of $H_0$  are stable against such perturbations. 
We then demonstrate that  locally block-diagonal perturbations satisfy the relative boundness condition
which proves the gap stability for   locally block-diagonal perturbations.

Note that our
condition of relative boundness is different from the one used in~\cite{yarot}. In particular,
our condition provides an elementary proof of the gap stability without the need to use cluster expansions
as was done in~\cite{yarot}.

Let $H_0$ and $W$ be any Hamiltonians acting on some Hilbert space $\calH$. We shall say that $W$ is {\em relatively bounded} by $H_0$
iff there exist $0\le b<1$  such that
\be
\label{rel_bound}
\| W \psi \| \le b\,  \| H_0 \psi\| \quad \mbox{for all $|\psi\ra \in \calH$}.
\ee

The following lemma asserts that a relatively bounded perturbation can change
eigenvalues of $H_0$ at most by a factor $1\pm b$.
\begin{lemma}
\label{lemma:relb}
Suppose $W$ is relatively bounded by $H_0$.
Then the spectrum of $H_0+W$ is contained in the union
of intervals $[\lambda_0(1-b),\lambda_0(1+b)]$ where $\lambda_0$
runs over the spectrum of $H_0$.
\end{lemma}
\begin{proof}
Indeed, suppose $(H_0+W)\, |\psi\ra = \lambda\, |\psi\ra$, that is,
\be
(H_0-\lambda \, I) \, |\psi\ra = -W\, |\psi\ra.
\ee
The relative boundness then implies $\| (H_0-\lambda \, I) \psi\| \le b \| H_0 \psi\|$, that is,
\be
\label{aux00}
\la \psi| (H_0-\lambda \, I)^2 |\psi\ra  \le b^2 \la \psi| H_0^2  |\psi\ra.
\ee
Let $H_0=\sum_{\lambda_0} \lambda_0 P_{\lambda_0}$ be the spectral decomposition of $H_0$.
Here the sum runs over the spectrum  of $H_0$ and $P_{\lambda_0}$ is a projector
onto the eigenspace with an  eigenvalue $\lambda_0$.
Define a probability distribution $p(\lambda_0)=\la \psi|P_{\lambda_0} |\psi\ra$.
Substituting it into Eq.~(\ref{aux00}) one gets
\be
\sum_{\lambda_0} (\lambda_0-\lambda)^2  \, p(\lambda_0) \le \sum_{\lambda_0} b^2 \lambda_0^2 \, p(\lambda_0).
\ee
Therefore there exists at least one eigenvalue $\lambda_0$ such that
\be
(\lambda_0-\lambda)^2 \le b^2 \lambda_0^2.
\ee
This is equivalent to $\lambda_0(1-b)\le \lambda\le \lambda_0(1+b)$.
\end{proof}

In the rest of this section we proof stability under locally block-diagonal perturbations.
Let $H_0$ be a Hamiltonian satisfying TQO-1,2.
\begin{lemma}
\label{lemma:lbd}
Let $V$ be a locally block-diagonal perturbation with strength $J$.
Then the spectrum of $H_0+V$ is contained in the union
of intervals  $\bigcup_{k\ge 0} I_k$, where $k$ runs over the spectrum of $H_0$
and
\[
I_k = \{ \lambda\in \RR\, : \, k(1-b)-\delta \le \lambda \le k(1+b)+\delta\}
\]
for some $b=O(J)$ and for some $\delta$ decaying faster than any power of $L^*$.
\end{lemma}
\begin{proof}
For any integer $r\ge 1$ define
\[
W(r)=\sum_{A\in \calS(r)} V_{r,A}.
\]
By assumptions of the lemma we have
\be
\label{decay1}
\|V_{r,A}\| \le J f(r)
\ee
for some function $f$ decaying faster than any power and
$[V_{r,A},P]=0$. Performing an overall energy shift and using TQO-1 we can assume that
\be
\label{annihilate}
V_{r,A}P=PV_{r,A}=0 \quad \mbox{for all $1\le r\le L^*$ and for all $A\in \calS(r)$}
\ee
\begin{prop}
Suppose $1\le r\le L^*$. Then
$W(r)$ is relatively bounded by $H_0$ with a constant $b(r)=O(J r^2 f(r))$.
\end{prop}
This proposition immediately implies that $\sum_{1\le r\le L^*} W(r)$ is relatively bounded by $H_0$
with a constant $b=\sum_{r=1}^{L^*} b(r)=O(J)$.
Lemma~\ref{lemma:relb} then implies that the spectrum of
$H_0+\sum_{1\le r\le L^*} W(r)$ is contained in the union of intervals
$[k(1-b),k(1+b)]$.
Treating the residual terms $\sum_{r>L^*} W(r)$
using the standard perturbation theory then leads to the desired result.
\begin{proof}[\bf Proof of the proposition]
Let $\Lambda=B_1\cup B_2 \cup \ldots \cup B_M$ be a partition of the lattice into contiguous squares of size $r$ such that any $2\times 2$ square is contained in exactly one square $B_a$.
If a $2\times 2$ square is on the boundary of two, or more, of the $B_a$, we include it in the bottom left one in the partition. Moreover, if $r$ does not divide $L$, then squares at the boundary may be truncated to rectangles, to fill in the partition.
We refer to $B_a$ as {\em boxes} to distinguish them from the squares involved in the decomposition of $V$.
For any binary string $\calY\in \{0,1\}^M$ define a projector
\[
R_\calY = \prod_{a=1}^M  \left[ \calY_a Q_{B_a}  + (1-\calY_a) P_{B_a}\right].
\]
Clearly the family of projectors $R_\calY$ defines an orthogonal decomposition
of the Hilbert space, that is, $\sum_\calY  R_\calY=I$.
Given a string $\calY\in \{0,1\}^M$, we shall say that a box $B_a$ is {\em occupied}
iff $\calY_a=1$.

We claim that
any operator $V_{r,A}$ acting on a square $A\in \calS(r)$ and  satisfying Eq.~(\ref{annihilate}) has only a few
off-diagonal blocks with respect to this decomposition. Specifically, TQO-2 implies that
\be
\label{trans}
R_\calY V_{r,A} R_\calZ\ne 0
\ee
only if $A$ has distance $O(1)$ from some occupied box in $\calY$ {\em and}
$A$ has distance $O(1)$ from some occupied box in $\calZ$, {\em and}
the configurations $\calY,\calZ$ differ only at those boxes that overlap with $A$.
Clearly, for any fixed $\calY$ such that $\calY$ has $k$ occupied boxes
the number of pairs $(A\in \calS(r), \calZ)$ that could satisfy Eq.~(\ref{trans})
is at most $O(kr^2)$.
Denote for simplicity
\[
W\equiv W(r), \quad w\equiv J f(r).
\]
For any state $|\psi\ra$ we get
\bea
\la \psi|W^2|\psi\ra &= &  \sum_{\calY,\calZ,\calV\subseteq [M]} \la \psi|R_\calY W R_\calZ W R_\calV |\psi\ra   \nn \\
&\le &
\sum_{\calY,\calZ,\calV} \| R_\calY W R_\calZ \| \cdot \| R_\calZ W R_\calV \| \cdot
\| R_\calY \psi \| \cdot \| R_\calV \psi\| \nn \\
&\le &
\sum_{\calY,\calZ,\calV} \| R_\calY W R_\calZ \| \cdot \| R_\calZ W R_\calV \| \cdot \frac12 (\la \psi|R_\calY |\psi\ra+
\la \psi|R_\calV|\psi\ra )\nn \\
&= & \sum_{\calY,\calZ,\calV}  \| R_\calY W R_\calZ \| \cdot \| R_\calZ W R_\calV \| \cdot \la \psi|R_\calY |\psi\ra \nn \\
& \le & \sum_{k\ge 0}\; \;  \sum_{\calY\, : \, |\calY|=k} \; \; O(k^2 r^4 w^2) \la \psi|R_\calY |\psi\ra = O(w^2 r^4 ) \la \psi|G|\psi\ra,
\label{W^2}
\eea
where
\be
G=\sum_{k\ge 0} \; \; \sum_{\calY\, : \, |\calY|=k} \; \; k^2 R_\calY.
\ee
The  inequality Eq.~(\ref{W^2}) follows from the fact that $\calY$ and $\calZ$ differ at at most $O(1)$ boxes
and an obvious bound $k(k+O(1))=O(k^2)$. Finally, note that $G\le H_0^2$ since
any  configuration with $k$ occupied boxes must have at least $k$ defects (squares $A\in \calS(2)$
for which $Q_A$ has eigenvalue $1$)
and since creating a defect costs at least a unit of energy.
We arrive at
\be
\la \psi|W^2|\psi\ra \le b^2 \la \psi|H_0^2 |\psi\ra, \quad b=O(wr^2).
\ee
This completes the proof.

\end{proof}
\end{proof}

\section{Exact quasi-adiabatic continuation}
\label{sec:continue}
We define a continuous family of Hamiltonians,
\be
H_s=H_0+sV,
\ee
so that as $s$ varies from $0$ to $1$, $H_s$ continuously interpolates between $H_0$ and the perturbed Hamiltonian.

We define a quasi-adiabatic continuation operator,
${\cal D}_s$ by
\be
i{\cal D}_s \equiv \int {\rm d}t F(t) \exp(i H_s t) \Bigl( \partial_s H_s \Bigr) \exp(-i H_s t),
\ee
where the function $F(t)$ is defined to have the following properties.
First, the Fourier transform of $F(t)$, which we denote $\tilde F(\omega)$, obeys
\be
\label{correct}
|\omega| \geq 1/2 \quad \rightarrow \quad \tilde F(\omega)=-1/\omega.
\ee
Second, $\tilde F(\omega)$ is infinitely differentiable, so that $F(t)$ decays faster than any power of $|t|$.
Third, $F(t)=-F(-t)$, so that ${\cal D}_s$ is Hermitian.

We define a unitary operator $U_s$ by
\be
\label{Usdef}
U_s \equiv {\cal S}' \; \exp(i\int_0^s {\rm d}s' {\cal D}_{s'} ),
\ee
where the notation ${\cal S}'$ denotes that the above equation (\ref{Usdef}) is an $s'$-ordered exponential.

The motivation for defining the above unitary operator is contained in the following lemma.
\begin{lemma}
\label{lemma:samespace}
Let  $H_s$ be a differentiable family of Hamiltonians.
Let
$|\Psi^i(s)\rangle$ denote eigenstates of $H_s$ with energies $E_i$.
Let $E_{min}(s)<E_{max}(s)$ be continuous functions of $s$.
Define a projector $P(s)$ onto an eigenspace of $H_s$ by
\be
P(s)=\sum_{i}^{E_i\in [E_{min}(s),E_{max}(s)]} |\Psi^i(s)\rangle \langle \Psi^i(s)|.
\ee
Assume that the
space that $P(s)$ projects onto is separated from the
rest of the spectrum by a gap of at least $1/2$ for all $s$ with $0\leq s \leq 1$.
That is, all eigenvalues of $H_s$ are either in the interval
$[E_{min}(s),E_{max}(s)]$, or are separated by at least $1/2$ from this interval.
Then, for all $s$ with $0\leq s \leq 1$, we have
\be
\label{result}
P(s)=U_s P(0) U_s^\dagger.
\ee
\begin{proof}
This is lemma 7.1 in \cite{prev}.  It is also well-known in the theory of quasi-adiabatic continuation.
\end{proof}
\end{lemma}

\section{Stability proof}
\label{sec:proof}

In this section we prove Theorem~\ref{thm:main}.
For any $s\in [0,1]$ define
\[
H_s=H_0+sV.
\]
Let $g$ be the ground state degeneracy of $H_0$. Choose $E_{min}(s)$ as the smallest eigenvalue
of $H_s$ and let $E_{max}(s)$ be the $g$-th smallest eigenvalue of $H_s$ (taking into account multiplicities).
Finally, let $\Delta(s)$ be the spectral gap separating eigenvalues of $H_s$
in the interval $[E_{min}(s),E_{max}(s)]$ from the rest of the spectrum.
We shall choose the constants $J_0,c_1$ sufficiently small such that the interval $I_0$
is separated from $I_k$, $k>0$ by a gap at least $3/4$.
Then the theorem implies that the interval $[E_{min}(s),E_{max}(s)]$ is contained in $I_0$
for all $s\in [0,1]$
and hence
\be
\label{34bound}
\Delta(s) \ge  \frac34 \quad \mbox{for all $0\le s\le 1$}.
\ee
Suppose we have already proved the theorem for the special case when
\be
\label{12bound}
\Delta(s)\ge \frac12 \quad \mbox{for all $s\in [0,1]$}.
\ee
In the remaining case there must exist $s^*\in [0,1)$ such that $\Delta(s)\ge 1/2$
for $s\in [0,s^*]$ and $\Delta(s^*) =1/2$ (use the fact that $\Delta(0)\ge 1$
and continuity of $\Delta(s)$).  Applying the theorem to a perturbation
$s^*V$ which satisfies Eq.~(\ref{12bound}) we conclude that $\Delta(s^*)\ge 3/4$
obtaining a contradiction.
Thus it suffices to prove the theorem for the case Eq.~(\ref{12bound}).

Define
\be
H_s'=U_s^\dag (H_0+sV) U_s=U_s^\dag H_s U_s,
\ee
where
\be
\label{Us}
U_s \equiv {\cal S}' \; \exp(i\int_0^s {\rm d}s' {\cal D}_{s'} ),
\ee
is the exact adiabatic continuation operator constructed in Section~\ref{sec:continue}.
Since we assumed $\Delta(s)\ge 1/2$ for all $s$, Lemma~\ref{lemma:samespace}
implies that
\be
\label{com1}
[H_s',P]=0.
\ee
We can represent $H_s'$ as
\be
H_s'=H_0+V',
\ee
where
\be
V'=U_s^\dag H_0 U_s - H_0 + sU_s^\dag VU_s.
\ee
Lemma~\ref{lemma:LRfast} implies that $\calD_s$ has strength $O(J)=O(1)$.
Applying Lemma~\ref{lemma:LRslow}  we conclude that $sU_s^\dag VU_s$
has strength $O(J)$.
Let us now focus on the term $U_s^\dag H_0 U_s - H_0$. We use
\be
U_s^\dagger H_0 U_s-H_0 = -i\int_0^s {\rm d}s'  \, U_{s'}^\dag [{\cal D}_{s'},H_0] U_{s'}.
\ee
Since $\calD_{s'}$ has strength $O(J)$,  the commutator $[{\cal D}_{s'},H_0]$
also has strength $O(J)$. Applying Lemma~\ref{lemma:LRslow} to the unitary evolution
$U_{s'}^\dag$, we infer that $U_s^\dagger H_0 U_s-H_0$ has strength $O(J)$.
To conclude, we have shown that $V'$ has strength $O(J)$, and Eq.~(\ref{com1}) implies
\be
[V',P]=0,
\ee
that is, $V'$ is a globally block-diagonal perturbation with strength $O(J)$.
In Lemma~\ref{rewrite1} below, we will show that we can rewrite
\be
\label{rw1}
H_s'=H_0+V'=
H_0 + \sum_{u \in \Lambda} X_u
\ee
where $X_u$ obeys $\quad [X_u, P]=0$,
where $X_u$ has strength $O(J)$ and its support is near $u$.
Applying
Lemma~\ref{lemma:solveloc} to Eq.~(\ref{rw1}) implies that $H_s'$ can be written in the form
\be
\label{lbddec}
H_s'=H_0+V'=H_0+V''+\Delta,
\ee
where $V''$ is locally block-diagonal with strength $O(J)$ and $\|\Delta\|$
decays faster than any power of $L^*$.
The statement of the theorem then follows straightforwardly from Lemma~\ref{lemma:lbd}
which asserts that $V''$ is relatively bounded by $H_0$ with a small error
decaying faster than any power of $L^*$.

The Lemma~\ref{rewrite1} that we need follows the idea in~\cite{kitaev} to write a Hamiltonian of a gapped system as a sum of terms such that the ground states are
eigenvectors of each term separately (in \cite{solvloc} a related idea of writing it so that the ground state was an approximate eigenvector
of each term separately was considered).  The properties of $H_s'$ that we use are that it is globally block diagonal, it has
a spectral gap $\geq 1/2$,
the perturbation $V'$ has strength $J$, and that it is unitarily related by $U_s$ to a Hamiltonian with a decay rate $\mu>0$.
\begin{lemma}
\label{rewrite1}
Let $H_s'$ be defined as above.  Then, we can re-write
\be
H_s'=H_0+ \sum_{u \in \Lambda} X_u
\ee
where $X_u$ obeys $\quad [X_u, P]=0$,
where $X_u$ has strength $O(J)$ and its support is near $u$.
\begin{proof}
We start from representing $V'$ as $V'=\sum_{u\in \Lambda} V_u$, where
$V_u$ includes only interactions affecting a site $u$. Then $V_u$ has strength $J$ and its support is near $u$.
We set
\be
\tilde{V}_u = \int_{-\infty}^\infty  dt \, g(t) \exp(i H_s' t) V_u \exp(-i H_s' t),
\ee
where $g(t)$ is a function satisfying $g(-t)=g(t)^*$  such that its Fourier transform $\tilde g(\omega)$ is infinitely differentiable, has $\tilde g(0)=1$, and
$\tilde g(\omega)=0$ for $|\omega|\geq 1/2$.
Define
\be
\tilde Q_A = \int_{-\infty}^\infty  dt\, g(t)  \exp(i H_s' t) Q_A \exp(-i H_s' t).
\ee
Then,
\be
H_s'=\int_{-\infty}^\infty dt\, g(t)  \exp(i H_s' t) H_s' \exp(-i H_s' t)=
\sum_{A\in \calS(2)} \tilde Q_A + \sum_{u\in \Lambda}  \tilde{V}_u.
\ee
By construction of $\tilde{g}(\omega)$ we have $(1-P) \tilde Q_A P = (1-P) \tilde{V}_u P =0$.
Hence both $\tilde{Q}_A$ and $\tilde{V}_u$ preserve $P$.
Using the definition of $\tilde{V}_u$ we get
\be
\label{tildeVu}
U_s \tilde{V}_u U_s^\dag  = \int_{-\infty}^\infty  dt \, g(t) \exp(i H_s t) U_s V_u U_s^\dag  \exp(-i H_s t).
\ee
Recall that $\calD_s$ has strength $O(1)$.
Thus we can apply Lemma~\ref{lemma:LRslow}
to the unitary evolution $U_s$ to infer that $U_s V_u U_s^\dag$ has strength $O(J)$.
Because $\tilde g(\omega)$ is infinitely differentiable, $g(t)$ decays faster than any power.
Also, by assumptions of the theorem,  $H_s$ has strength $O(1)$
and decay rate $\mu>0$.
Hence we can apply Lemma~\ref{lemma:LRfast}
to the unitary evolution $\exp(i H_s t)$ to infer that $U_s \tilde{V}_u U_s^\dag$ has strength $O(J)$,
Finally, applying Lemma~\ref{lemma:LRslow} to the unitary evolution $U_s^\dag$
we infer that $\tilde{V}_u$ has strength $O(J)$.
In addition, $\tilde{V}_u$ has support near $u$ since all Hamiltonians obtained at the intermediate steps
have support near $u$, see Lemmas~\ref{lemma:LRslow},\ref{lemma:LRfast}.

We now consider the terms $\tilde Q_A$.  We have
\begin{eqnarray}
\tilde Q_A & = & \int_{-\infty}^{\infty} dt\,g(t) \exp(i H_s' t) Q_A \exp(-i H_s' t)  \\ \nonumber
&=& Q_A+ i\int_{-\infty}^{\infty} dt\, g(t) \int_0^t d t_1\,  \exp{(i H_s' t_1 )} [V,Q_A] \exp{(-iH_s' t_1)} \\ \nonumber
&=& Q_A+ \int_{-\infty}^{\infty} dt_1\, f(t_1)  \exp{(i H_s' t_1)} [V,Q_A] \exp{(-iH_s' t_1)} \nn
\end{eqnarray}
for some function $f(t_1)$ decaying faster than any power.
It follows that
\be
U_s (\tilde{Q}_A - Q_A)U_s^\dag = \int_{-\infty}^{\infty} dt_1\, f(t_1)  \exp{(i H_s t_1)} U_s [V,Q_A] U_s^\dag  \exp{(-iH_s t_1)}.
\ee
Applying the same arguments as in the analysis of Eq.~(\ref{tildeVu})
we conclude that  $\tilde{Q}_A-Q_A$
has strength $O(J)$ and has support near $u(A)$ --- the center of the square $A$.
In addition, $\tilde{Q}_A-Q_A$ commutes with $P$ since both
$\tilde{Q}_A$ and $Q_A$ do. Let us define
$X_u=\tilde{V}_u + (\tilde{Q}_A-Q_A)$ where $A$ is the square centered at $u$.
\end{proof}
\end{lemma}

{\it Acknowledgments---} We thank S. Michalakis for useful discussions and for collaboration on \cite{prev}.
SB was partially  supported by the
DARPA QUEST program under contract number HR0011-09-C-0047.

\end{document}